\documentclass[a4paper,10pt]{amsart}
\usepackage[utf8]{inputenc}
\usepackage{fullpage}
\usepackage{graphicx}
\usepackage{amssymb,amsmath,latexsym}
\usepackage{tikz}
\usepackage{color}
\usepackage{feynmp,epsf}
\usepackage{url}

\newcommand{\m}[1]{\mathcal{#1}}
\newcommand{\mb}[1]{\mathbb {#1}}

\newcommand{\ti}[1]{\widetilde {#1}}

\theoremstyle{plain}
\newtheorem{thm}{Theorem}[section]
\newtheorem*{thm*}{\bf Theorem }
\newtheorem{lem}[thm]{Lemma}
\newtheorem{prop}[thm]{Proposition}
\newtheorem*{prop*}{\bf Proposition}

\theoremstyle{definition}
\newtheorem{defn}[thm]{Definition}

\newtheorem{example}[thm]{Example}

\theoremstyle{remark}
\newtheorem{remark}[thm]{Remark}

\numberwithin{equation}{section}

\title{Complexes of marked graphs in gauge theory}
\author{Marko Berghoff and Andre Knispel}

\begin{document}

\begin{abstract}
We review the gauge and ghost cyle graph complexes as defined by Kreimer, Sars and van Suijlekom in ``Quantization of gauge fields, graph polynomials and graph homology'' and compute their cohomology. These complexes are generated by labelings on the edges or cycles of graphs and the differentials act by exchanging these labels. We show that both cases are instances of a more general construction of double complexes associated to graphs. Furthermore, we describe a universal model for these kind of complexes which allows to treat all of them in a unified way.
\end{abstract}

\maketitle

\section{Introduction}
\label{intro}
In \cite{ksvs} Kreimer, Sars and van Suijlekom showed how gauge theory amplitudes can be generated using only a scalar field theory with cubic interaction. On the analytic side this is achieved by means of a new graph polynomial, dubbed the \textit{corolla polynomial}, that transforms integrands of scalar graphs into gauge theory integrands. On the combinatorial side all graphs relevant in gauge theory can be generated from the set of all 3-regular graphs by means of operators that label edges and cycles. These labels represent edges with different Feynman rules that incorporate contributions from 4-valent vertices and relations between 3- and 4-valent vertices, and similar for gluon and ghost cycles. 

Generating and exchanging these labels on a fixed graph $\Gamma$ can be cast as operations that square to zero, hence define differentials on the free abelian group generated by all possible labelings of $\Gamma$. One of the main observations in \cite{ksvs} is that modeling edge-collapses and particle types by different labels on edges and cycles, called \textit{markings}, one thereby obtains two cochain complexes, called \textit{gauge} and \textit{ghost cycle complexes}, whose cohomology encodes physical constraints on scattering amplitudes in gauge theory. \textit{Very} roughly speaking, the first marking represents modified Feynman rules, such that that the full gauge theory amplitude is given by the sum over all marked, 3-regular graphs (representing all ways of expanding 4-gluon into 3-gluon vertices or all ways of exchanging gluon for ghost loops, respectively). The second marking or, more precisely, the two differentials that change the first into the second marking and generate new marked edges of the second type, reflect physical constraints such as unitarity and gauge covariance, in the sense that observable quantities must lie in the kernel of these maps (similar to the approach in BRST quantization, see for instance \cite{Ofarrill}). Thus the relevance of understanding the cohomology of these complexes. 

For a thorough discussion of the quantum field theoretical motivation and interpretation of these complexes we refer to the original article \cite{ksvs} and the review \cite{DK-corolla}. A detailed discussion of the analytic approach via corolla polynomials can be found in \cite{sars}, a general reference for background material on the qauntization of gauge theories is the classical work \cite{cvitanovic}. 

However, in the present article we are not concerned with physics, but only with the cohomology of these complexes. In \cite{ksvs} it is stated that the gauge theory amplitude is a cocycle in both complexes. The authors then discuss the physical implications of this fact. Here we study the full cohomology of these complexes and show that this amplitude is not only a cocycle, but represents a non-trivial cohomology class, in fact the only one. 

We show that the two complexes introduced in \cite{ksvs} are special cases of a general construction that associates a cochain complex to a graph and a class of subgraphs allowed to be marked. This complex is generated by all possible markings of the graph and the differentials operate on the markings by generating and exchanging them. The connection to physics comes here from the mere choice of marked substructures (i.e.\ cycles and edges) and the interpretation of the differentials.

Note that we are not dealing with ``classical'' graph complexes in the sense of \cite{ko2}. Although edge-markings may be interpreted as Feynman rules for edge-collapses, the differentials do not change the topology of graphs. Our construction is more similar to \cite{jonsson} which studies simplicial complexes associated to graphs and classes of substructures, such as cliques and independent sets. 

Our main statement is the following.

\begin{thm}\label{thm:mainthm} Fix $r,l\in \mb N$ and let $(\m G^\bullet,S+(-1)^\bullet T)$ denote the total complex where $S$ and $T$ are the gauge and ghost cycle differentials.
Define
 \begin{equation*}
  X:= \sum_{\Gamma} \sum_{m} (\Gamma,m),
 \end{equation*}
as the sum over all admissible 1-markings of edges and cycles in 3-regular graphs $\Gamma$ with $r$ legs and $l$ loops. Then $SX=TX=0$ and $X$ represents the only non-trivial cohomology class in $H^\bullet(\m G,S+(-1)^\bullet T)$.
\end{thm}

This result is based on two properties of the gauge and ghost cycle complexes. Firstly, there is a universal model for general complexes of marked graphs allowing to treat both cases at once. Secondly, its differentials are of the form $D=\delta + d$ with $\delta$ very simple. This allows to compute the cohomology of $D$ by a spectral sequence argument without the need of explicitly understanding the cohomology with respect to $d$.\footnote{The differential $d$ is quite interesting in its own right as its computation is related to NP-hard problems in graph theory, cf.\ \cite{knispel}.} Both properties are established in theorems \ref{theorem:universal} and \ref{theorem:HomDC} below.

The exposition is organized as follows. In the next section we introduce markings and complexes of marked graphs in general, then we specialize to edge-, cycle- or vertex marked graphs, the latter serving as our ``computational model'' to study the former two cases. We compute its cohomology in Section \ref{section:ggcohom} in two steps. First we study only the cohomology with respect to the simple differential $\delta$, then we apply the result in a spectral sequence associated to the double complex formed by $\delta+ d$. In Section \ref{section:gaugegraphcomplex} we combine the gauge and ghost cycle complexes into a large double complex and show that its cohomology is generated by a single element, the full gauge theory amplitude\footnote{For the experts: Symmetry factors can of course be included (as shown in \cite{ksvs}), but they do not play a role for the cohomology of the complex.}, thereby proving the main theorem. 
\newline

\textbf{Acknowledgements:}
We are very grateful to Henry Kißler and Dirk Kreimer for many valuable discussions and explanations, especially about the physics behind this work.

\section{Complexes of marked graphs}\label{section:gaugegc}

We introduce some notation, then describe complexes of marked graphs of which the gauge and ghost cycle complexes in \cite{ksvs} emerge as special cases. We call them \textit{edge-} and \textit{cycle-marking complexes}. After discussing these two in detail we introduce the case of vertex-markings which serves as a universal model to study these kind of complexes.

\subsection{General markings}\label{subsection:markedcomplex}

Let $\Gamma=(\Gamma^0,\Gamma^1)$ be a connected graph. We call edges connected to univalent vertices \textit{external (edges)} or \textit{legs}, all other edges are referred to as \textit{internal} or, by abuse of language, simply as \textit{edges}. Thus, the set $\Gamma^1$ of edges of $\Gamma$ splits into $\Gamma^1= \Gamma^1_{ext} \sqcup \Gamma^1_{int}$.

A similar decomposition holds for the set of vertices of $\Gamma$, $\Gamma^0= \Gamma^0_{ext} \sqcup \Gamma^0_{int}$.

Since our operations will focus solely on the internal structure of graphs, we write $V=V(\Gamma):=\Gamma^0_{int}$ and $E=E(\Gamma):=\Gamma^1_{int}$ for its internal vertices and edges. In that spirit we denote graphs by $\Gamma=(V,E)$, as is customary in graph theory, tacitly remembering the external structure encoded by the pair $\left(\Gamma^0_{ext},\Gamma^1_{ext} \right) $.

A \textit{subgraph} of $\Gamma$ is a pair of subsets $V'\subset V$ and $ E' \subset E$ such that $(V',E')$ is a graph itself. Note that our definition of subgraphs relies only on the internal structure of $\Gamma$. However, we allow subgraphs to have univalent vertices. In particular, (internal) vertices and edges of $\Gamma$ are (identified with their corresponding) subgraphs of $\Gamma$.

\begin{defn}\label{definition:Smarking}
Fix a finite set $S = \{ 0, \ldots, s \} \subset \mb N$. For a finite graph $\Gamma$ let $\mathrm{Sub}(\Gamma)$ denote the set of all subgraphs of $\Gamma$. Given a subset $P\subset \mathrm{Sub}(\Gamma)$ a $P$-\textit{marking} of $\Gamma$ (\textit{in $S$}) is a map $m: P \to S$. We call the pair $(\Gamma,m)$ a \textit{marked graph} and think of the subgraphs in $m^{-1}(0)$ as being not marked. Likewise, for $i \in \{1,\ldots,s\}$ we refer to the elements of $m^{-1}(i)$ as $i$-marked. A marking is \textit{admissible} if no two marked elements share a common vertex and, in the case that vertices are marked, no two marked vertices are connected by an edge.
\end{defn}

In the following we will describe (co-)chain complexes of marked graphs. For this we consider only admissible markings with $s=2$, but more general settings are obviously possible. Throughout this work all coefficients will be in $\mb Z$.

\begin{defn}\label{definition:markedcomplex}
Fix a graph $\Gamma$ and $P\subset \mathrm{Sub}(\Gamma)$ a set of subgraphs of $\Gamma$ endowed with a total order. Let $\m P(\Gamma)$ denote the free abelian group generated by all markings $(\Gamma,m)$ where $m:P \to \{0,1,2\}$ is admissible. 

The marking induces a partition of $P$. We write $P=P_0 \sqcup P_m$ where $P_0$ denotes the unmarked objects in $P$ and $P_m=P_1 \sqcup P_2$ the 1- and 2-marked ones. 
\end{defn}

The group $\m P(\Gamma)$ carries two gradings. With $\m P(\Gamma)_i^j$ denoting the subgroup of $\m P(\Gamma)$ generated by markings with $|P_1|=i$ and $|P_2|=j$ we set
\begin{equation*}
 \m P(\Gamma)^j:= \bigoplus_{i \in \mb N}\m P(\Gamma)_i^j.
\end{equation*}

We define two differentials on this group by changing and permuting the markings.

\begin{defn}\label{definition:d1d2}
 For $(\Gamma,m) \in \m P(\Gamma)$ let $(\Gamma,m_{|p\mapsto 2})$ denote the marking that is identical to $m$ on $P\setminus p$ and marks $p$ by 2.
 Define linear maps $\delta,d: \m P(\Gamma)^j \longrightarrow \m P(\Gamma)^{j+1}$ by  
 
\begin{align*}
  \delta(\Gamma,m)&:= (-1)^{| P_m |}  \sum_{p \in P_1 } (-1)^{|\{ p' \in P_1  \mid p' > p\}|} \delta^p(\Gamma,m), \\
  d(\Gamma,m)&:= \sum_{p \in P_0 } (-1)^{|\{ p' \in P_m \mid p' < p\}|} d^p(\Gamma,m),
 \end{align*}
 where $\delta^p(\Gamma,m):=(\Gamma,m_{|p\mapsto 2})$ and
\begin{equation*}
 d^p(\Gamma,m):= \begin{cases}
           0 & \text{if $p$ shares a vertex with some $p'\in P$} \\
           (\Gamma,m_{|p\mapsto 2}) & \text{else.}
          \end{cases}
\end{equation*}
\end{defn}

\begin{remark}
 Obviously, definitions \ref{definition:markedcomplex} and \ref{definition:d1d2} depend on the chosen order on $P$. We omit this piece of data in the following, because the cohomology of these complexes will turn out to be independent of it. This can be shown explicitly, but follows also a posteriori from Proposition \ref{prop:cocycles}. See also Remark \ref{remark:edgeorder}.
\end{remark}

\begin{prop}\label{proposition:Ddifferential}
 Both maps $\delta$ and $d$ square to zero. Moreover, $\delta d+ d \delta=0$, so that $D:=\delta + d$ is a differential on $\m P(\Gamma)$.
\end{prop}

\begin{proof}
This is shown in \cite{ksvs}, propositions 4.14 and 4.17, for the case $P=E(\Gamma)$. The same proof works for general $P$, since due to the notion of admissible markings $\delta$ and $d$ cannot ``mix'' elements of $P$. Therefore,
\begin{align*}
  \delta \delta(\Gamma,m)&=  (-1)^{|P_m|} \sum_{p \in P_1 }(-1)^{|\{ p' \in P_1 \mid p' > p\}|}\delta \delta^p(\Gamma,m) \\
  &=  \sum_{p \in P_1 }(-1)^{|\{ p' \in P_1 \mid p' > p\}|} \sum_{q \in P_1 \setminus \{p\} }(-1)^{|\{ q' \in P_1 \setminus \{p\}  \mid q' > q\}|} \delta^q \delta^p (\Gamma,m) \\
  &=  \sum_{p \in P_1 }\sum_{q \in P_1 \setminus \{p\} }(-1)^{ |\{ p' \in P_1 \mid p' > p\}| +|\{ q' \in P_1 \setminus \{p\}  \mid q' > q\}|} \delta^q \delta^p (\Gamma,m) \\
  & = \sum_{p,q \in P_1, p<q }(-1)^{|\{ r \in P_m \mid p< r < q\}|} \delta^q \delta^p (\Gamma,m) \\
  & \quad +  \sum_{p,q \in P_1,q<p }(-1)^{|\{ r \in P_m \mid q< r < p\}| -1} \delta^q \delta^p (\Gamma,m) = 0.
 \end{align*}
  Similarly,
 \begin{align*}
  d d(\Gamma,m)&= \sum_{p \in P_0 }(-1)^{|\{ p' \in P_m \mid p' < p\}|} dd^p(\Gamma,m) \\
  &=  \sum_{p \in P_0 }(-1)^{|\{ p' \in P_m \mid p' < p\}|} \sum_{q \in P_0\setminus \{p\} }(-1)^{|\{ q' \in P_m \cup \{p\} \mid q' < q\}|} d^q d^p (\Gamma,m) \\
  &=  \sum_{p \in P_0 }\sum_{q \in P_0 \setminus \{p\}}(-1)^{|\{ p' \in P_m \mid p' < p\}| + |\{ q' \in P_m \cup \{p\} \mid q' < q\}|} d^q d^p (\Gamma,m) \\
  & = \sum_{p,q \in P_0, p<q }(-1)^{1+|\{ r \in P_m \mid p< r < q\}|} d^q d^p (\Gamma,m) \\
  & \quad +  \sum_{p,q \in P_0,q<p }(-1)^{|\{ r \in P_m \mid q< r < p\}| } d^q d^p (\Gamma,m) = 0.
 \end{align*}
Finally,
\begin{align*}
  d \delta(\Gamma,m)&=  (-1)^{|P_m|} \sum_{p \in P_1 }(-1)^{|\{ p' \in P_1 \mid p' > p\}|}d \delta^p(\Gamma,m) \\
  &= (-1)^{| P_m |} \sum_{p \in P_1, q \in P_0 }(-1)^{|\{ p' \in P_1 \mid p' > p\}|+|\{ q' \in P_m  \mid q' < q\}|} d^q \delta^p (\Gamma,m), \\
  \delta d (\Gamma,m) &= \sum_{p \in P_0 } (-1)^{|\{ p' \in P_m \mid p' < p\}|} \delta d^p(\Gamma,m) \\
  & = (-1)^{| P_m |+1} \sum_{p \in P_0,q \in P_1 } (-1)^{|\{ p' \in P_m \mid p' < p\}| + |\{ q' \in P_1  \mid q' > q\}|}
  \delta^q d^p(\Gamma,m) \\
  & = - d \delta(\Gamma,m).
 \end{align*}
\end{proof}

In summary, given a graph $\Gamma$ and any ``type'' of subgraphs $P \subset \mathrm{Sub}(\Gamma)$ that can be marked we get a cochain complex $(\m P(\Gamma),D)$. In the following we specialize this construction to three cases where $P$ consists of edges, cycles (the two cases considered in \cite{ksvs}) or vertices (our universal model).

\subsection{The edge-marking complex}\label{subsection:ggc}
We start with the case of edge-markings, i.e.\ we apply the construction outlined above to $P=E=E(\Gamma)$.
\begin{defn}\label{definition:Gra}
Let $\mathrm{Gra}_{r,l}$ denote the set of all isomorphism classes of finite connected graphs $\Gamma=(\Gamma^0,\Gamma^1)\approx (V,E)$ with the following properties
\begin{itemize}
 \item $\Gamma$ has first Betti number is equal to $l$, i.e.\ has $l$ independent cycles.
 \item $\Gamma$ has $r$ legs, i.e.\ $r$ vertices of valence one, all other vertices have valence equal to three.
 \item $\Gamma$ has no self-loops (edges connecting a single vertex to itself).
\end{itemize}
\end{defn}
\begin{defn}\label{definition:Ecomplex}
Fix $r,l \in \mb N$. For $\Gamma \in \mathrm{Gra}_{r,l}$ let $\m E(\Gamma)$ denote the free abelian group generated by all admissible markings $m:E \to \{0,1,2\}$ of the (ordered) edges of $\Gamma$. 

A marking induces a partition of the edge set. We write $E=E_0 \sqcup E_m$ where $E_0$ denotes the unmarked edges and $E_m=E_1 \sqcup E_2$ the 1- and 2-marked ones. 
\end{defn}

 \begin{remark}\label{remark:edgeorder}
  As mentioned in the previous section, all choices of orders on the marked elements produce isomorphic complexes. Nevertheless, a priori one needs to take care in regards to graph automorphisms and how they change a chosen order. In the usual definition of differentials on graph complexes graphs are oriented by an order on their edge set.\footnote{This is not the only choice; cf.\ \cite{convogt} for a thorough discussion of this matter.} Two orientations $o$ and $o'$ on a graph $\Gamma$ are related by
  \begin{equation*}
  (\Gamma,o)= \mathrm{sgn}(\varphi) \cdot (\Gamma,o'),
 \end{equation*}
  where $\varphi$ is the permutation on $E$ induced by the change of order. As a consequence, we have $(\Gamma,o)=-(\Gamma,o)=0$ for any graph that has an automorphisms inducing an odd permutation of its edge set. In particular, graphs with multi-edges vanish.
To keep all graphs in the game, we take the order as an additional, separately chosen piece of information, tacitly equipping every isomorphism class of graphs with a choice. Ultimately, this works because our differentials do not relate different graphs but operate on the marking only.
 \end{remark}

The group $\m E(\Gamma)$ carries two gradings. With $\m E(\Gamma)_i^j$ denoting the subgroup of $\m E(\Gamma)$ generated by marked graphs with $i$ edges of type 1 and $j$ edges of type 2 we set
\begin{equation*}
 \m E(\Gamma)^j:= \bigoplus_{i \in \mb N} \m E(\Gamma)_i^j.
\end{equation*}
Definition \ref{definition:d1d2} and Proposition \ref{proposition:Ddifferential} produce three differentials $\sigma,s, S: \m E(\Gamma)^j \longrightarrow \m E(\Gamma)^{j+1}$, given by
 \begin{align*}
  \sigma(\Gamma,m) & = (-1)^{| E_m |}  \sum_{e \in E_1 } (-1)^{|\{ e' \in E_1  \mid e' > e\}|} \sigma_e(\Gamma,m) ,\\
  s(\Gamma,m) & = \sum_{e \in E_0 } (-1)^{|\{ e' \in E_m \mid e' < e\}|} s_e(\Gamma,m),\\
 S & =s + \sigma,
 \end{align*}
 where $\sigma_e(\Gamma,m)= (\Gamma,m_{|e\mapsto 2})$ and 
 \begin{equation*}
 s_e(\Gamma,m)= \begin{cases}
           0 & \text{ if $e$ is adjacent to another marked edge} \\
           (\Gamma,m_{|e\mapsto 2}) & \text{ else.}
          \end{cases}    
 \end{equation*}        
          
\begin{example}\label{example:edgemarking}
Let $(\Gamma,m)=$
\raisebox{-.33cm}{\resizebox{2cm}{1cm}{\begin{tikzpicture}[scale=1]
 \draw []  (-.3,0) to (0,0);
\draw [out=90, in=90, looseness=1.3]  (0,0) to node[above, xshift=-.06cm] {$1$} node[xshift=.06cm] {$\boldsymbol{|}$} (1,0);
 \draw [out=-90, in=-90, , looseness=1.3]  (0,0) to node[below] {$2$} (1,0);
 \draw []  (1,0) to node[above] {$3$} (1.5,0);
 \draw [out=90, in=90, looseness=1.3]  (1.5,0) to node[above] {$4$} (2.5,0);
 \draw [out=-90, in=-90, looseness=1.3]  (1.5,0) to node[below] {$5$} (2.5,0);
 \draw []  (2.5,0) to (2.8,0);
 \end{tikzpicture}}}
with $E$ ordered as pictured, the 1- and 2-markings denoted by ``$|$'' and ``$||$'', respectively. Then

\begin{align*}
  s(\Gamma,m) & =- \raisebox{-.3cm}{\resizebox{2cm}{.8cm}{\begin{tikzpicture}[scale=1]
 \draw []  (-.3,0) to (0,0);
\draw [out=90, in=90, looseness=1.3]  (0,0) to node {\textbf{$\boldsymbol{|}$}} (1,0);
 \draw [out=-90, in=-90, looseness=1.3]  (0,0) to node[below] {} (1,0);
 \draw []  (1,0) to node[above] {} (1.5,0);
 \draw [out=90, in=90, looseness=1.3]  (1.5,0) to node[] {\textbf{$\boldsymbol{||}$}} (2.5,0);
 \draw [, out=-90, in=-90, looseness=1.3]  (1.5,0) to node[] {} (2.5,0);
 \draw []  (2.5,0) to (2.8,0);
 \end{tikzpicture}}}
-
\raisebox{-.3cm}{\resizebox{2cm}{.8cm}{\begin{tikzpicture}[scale=1]
 \draw []  (-.3,0) to (0,0);
\draw [out=90, in=90, looseness=1.3]  (0,0) to node {\textbf{$\boldsymbol{|}$}} (1,0);
 \draw [out=-90, in=-90, looseness=1.3]  (0,0) to node[below] {} (1,0);
 \draw []  (1,0) to node[above] {} (1.5,0);
 \draw [out=90, in=90, looseness=1.3]  (1.5,0) to node[] {} (2.5,0);
 \draw [, out=-90, in=-90, looseness=1.3]  (1.5,0) to node[] {\textbf{$\boldsymbol{||}$}} (2.5,0);
 \draw []  (2.5,0) to (2.8,0);
 \end{tikzpicture}}}
, \quad
\sigma(\Gamma,m) = - 
\raisebox{-.3cm}{\resizebox{2cm}{.8cm}{\begin{tikzpicture}[scale=1]
 \draw []  (-.3,0) to (0,0);
\draw [out=90, in=90, looseness=1.3]  (0,0) to node {\textbf{$\boldsymbol{||}$}} (1,0);
 \draw [out=-90, in=-90, looseness=1.3]  (0,0) to node[below] {} (1,0);
 \draw []  (1,0) to node[above] {} (1.5,0);
 \draw [out=90, in=90, looseness=1.3]  (1.5,0) to node[above] {} (2.5,0);
 \draw [, out=-90, in=-90, looseness=1.3]  (1.5,0) to node[below] {} (2.5,0);
 \draw []  (2.5,0) to (2.8,0);
 \end{tikzpicture}}}, \\
\sigma s (\Gamma,m) & = -
\raisebox{-.3cm}{\resizebox{2cm}{.8cm}{\begin{tikzpicture}[scale=1]
 \draw []  (-.3,0) to (0,0);
\draw [out=90, in=90, looseness=1.3]  (0,0) to node {\textbf{$\boldsymbol{||}$}} (1,0);
 \draw [out=-90, in=-90, looseness=1.3]  (0,0) to node[below] {} (1,0);
 \draw []  (1,0) to node[above] {} (1.5,0);
 \draw [out=90, in=90, looseness=1.3]  (1.5,0) to node[] {\textbf{$\boldsymbol{||}$}} (2.5,0);
 \draw [, out=-90, in=-90, looseness=1.3]  (1.5,0) to node[below] {} (2.5,0);
 \draw []  (2.5,0) to (2.8,0);
 \end{tikzpicture}}} 
- \raisebox{-.3cm}{\resizebox{2cm}{.8cm}{\begin{tikzpicture}[scale=1]
 \draw []  (-.3,0) to (0,0);
\draw [out=90, in=90, looseness=1.3]  (0,0) to node {\textbf{$\boldsymbol{||}$}} (1,0);
 \draw [out=-90, in=-90, looseness=1.3]  (0,0) to node[below] {} (1,0);
 \draw []  (1,0) to node[above] {} (1.5,0);
 \draw [out=90, in=90, looseness=1.3]  (1.5,0) to node[] {} (2.5,0);
 \draw [, out=-90, in=-90, looseness=1.3]  (1.5,0) to node[] {\textbf{$\boldsymbol{||}$}} (2.5,0);
 \draw []  (2.5,0) to (2.8,0);
 \end{tikzpicture}}} 
 = - s \sigma (\Gamma,m).
\end{align*}
\end{example}
 
Let $\m E:= \bigoplus_{\Gamma \in \mathrm{Gra}_{r,l}} \m E(\Gamma)$. This group naturally inherits a grading and a differential $S=s+ \sigma$ from each of its summands, hence defines a cochain complex. 

\begin{defn}\label{definition:gluongraphcomplex}
 The complex $(\m E,S)$ is called the \textit{edge-marking complex}.
\end{defn}

\subsection{The cycle-marking complex}\label{subsection:gcgc}
Now we mark cycles instead of edges. 

\begin{defn}\label{definition:cycle}
Let $\Gamma$ be a connected graph. A \textit{cycle} $c$ in $\Gamma$ is a closed path without repeated vertices, i.e.\ a subgraph $c \subset \Gamma$ such that 
\begin{itemize}
 \item every $v\in V$ is incident to none or exactly two elements in $c$.
 \item $c \subset \Gamma$ is connected.
\end{itemize}
We denote by $C=C(\Gamma)$ the set of all cycles in $\Gamma$.
\end{defn}

\begin{defn}\label{definition:Ccomplex}
For fixed $r,l \in \mb N$ and $\Gamma \in \mathrm{Gra}_{r,l}$ let $\m C(\Gamma)$ denote the free abelian group generated by all admissible markings $m:C \to \{0,1,2\}$ of the (ordered set of) cycles of $\Gamma$. 

Analogous to the case of edge-markings, every cycle-marking induces a partition of the cycle set, $C=C_0 \sqcup C_m$ and $C_m=C_1 \sqcup C_2$.
\end{defn}

The group $\m C(\Gamma)$ is bigraded by the number of 1- and 2-markings. Let $\m C(\Gamma)_i^j$ denote the subgroup generated by marked graphs with $i$ cycles of type 1 and $j$ cycles of type 2 and let
\begin{equation*}
 \m C(\Gamma)^j:= \bigoplus_{i \in \mb N} \m C(\Gamma)_i^j.
\end{equation*}

The three differentials $\delta,d,D$ translate to $\tau,t, T: \m C(\Gamma)^j \longrightarrow \m C(\Gamma)^{j+1}$:
 \begin{align*}
  \tau(\Gamma,m)&= (-1)^{| C_m |}  \sum_{c \in C_1 } (-1)^{|\{ c' \in C_1  \mid c' > c\}|} \tau_c(\Gamma,m),\\
  t(\Gamma,m)&= \sum_{c \in C_0 } (-1)^{|\{ c' \in C_m \mid c' < c\}|} t_c(\Gamma,m), \\
  T & =t + \tau,
 \end{align*}
 where $\tau_c(\Gamma,m)= (\Gamma,m_{|c\mapsto 2})$ and 
 \begin{equation*}
  t_c(\Gamma,m)= \begin{cases}
           0 & \text{ if $c$ shares a vertex with another marked cycle} \\
           (\Gamma,m_{|c\mapsto 2}) & \text{ else.}
          \end{cases}
 \end{equation*}

\begin{example}\label{example:cyclemarking}
Let $(\Gamma,m)=$
\raisebox{-.23cm}{\resizebox{2cm}{.8cm}{\begin{tikzpicture}[scale=1]
 \draw []  (-.3,0) to (0,0);
\draw [out=90, in=90, looseness=1.3]  (0,0) to node[] {} node[] {} (1,0);
 \draw [out=-90, in=-90, , looseness=1.3]  (0,0) to node[above, yshift=.17cm] {$1$} node[] {}  (1,0);
 \draw []  (1,0) to node[above] {} (1.5,0);
 \draw [dotted, out=90, in=90, looseness=1.3]  (1.5,0) to node[] {} (2.5,0);
 \draw [dotted, out=-90, in=-90, looseness=1.3]  (1.5,0) to node[above,yshift=.17cm] {$2$} (2.5,0);
 \draw []  (2.5,0) to (2.8,0);
 \end{tikzpicture}}}
with $C$ ordered as pictured, the 1- and 2-markings drawn as dotted and dashed cycles, respectively. Then
\begin{align*}
  t(\Gamma,m) & =
\raisebox{-.3cm}{\resizebox{2cm}{.8cm}{\begin{tikzpicture}[scale=1]
 \draw []  (-.3,0) to (0,0);
\draw [dashed, out=90, in=90, looseness=1.3]  (0,0) to node[above, xshift=-.05cm] {} (1,0);
 \draw [dashed, out=-90, in=-90, , looseness=1.3]  (0,0) to node[below] {} (1,0);
 \draw []  (1,0) to node[above] {} (1.5,0);
 \draw [dotted, out=90, in=90, looseness=1.3]  (1.5,0) to node[above] {} (2.5,0);
 \draw [dotted, out=-90, in=-90, looseness=1.3]  (1.5,0) to node[below] {} (2.5,0);
 \draw []  (2.5,0) to (2.8,0);
 \end{tikzpicture}}}
, \quad
\tau(\Gamma,m) = - 
\raisebox{-.3cm}{\resizebox{2cm}{.8cm}{\begin{tikzpicture}[scale=1]
 \draw []  (-.3,0) to (0,0);
\draw [out=90, in=90, looseness=1.3]  (0,0) to node[above, xshift=-.05cm] {} (1,0);
 \draw [out=-90, in=-90, , looseness=1.3]  (0,0) to node[below] {} (1,0);
 \draw []  (1,0) to node[above] {} (1.5,0);
 \draw [dashed, out=90, in=90, looseness=1.3]  (1.5,0) to node[above] {} (2.5,0);
 \draw [dashed, out=-90, in=-90, looseness=1.3]  (1.5,0) to node[below] {} (2.5,0);
 \draw []  (2.5,0) to (2.8,0);
 \end{tikzpicture}}}, \\
\tau t (\Gamma,m) & = 
\raisebox{-.3cm}{\resizebox{2cm}{.8cm}{\begin{tikzpicture}[scale=1]
 \draw []  (-.3,0) to (0,0);
\draw [dashed, out=90, in=90, looseness=1.3]  (0,0) to node[above, xshift=-.05cm] {} (1,0);
 \draw [dashed, out=-90, in=-90, , looseness=1.3]  (0,0) to node[below] {} (1,0);
 \draw []  (1,0) to node[above] {} (1.5,0);
 \draw [dashed, out=90, in=90, looseness=1.3]  (1.5,0) to node[above] {} (2.5,0);
 \draw [dashed, out=-90, in=-90, looseness=1.3]  (1.5,0) to node[below] {} (2.5,0);
 \draw []  (2.5,0) to (2.8,0);
 \end{tikzpicture}}} = - t \tau(\Gamma,m).
\end{align*}
\end{example}

\begin{defn}\label{definition:ghostcyclecomplex}
Let $\m C:= \bigoplus_{\Gamma \in \mathrm{Gra}_{r,l}} \m C(\Gamma)$, graded by the number of 1- and 2-marked cycles, and equipped with the differential $T=t+ \tau$. The complex $(\m C,T)$ is called the \textit{cycle-marking complex}.
\end{defn}

\subsection{Marking vertices}\label{subsection:markingvertices}
Note that all of the previously defined differentials do not alter the topology of graphs, they only change their markings. Furthermore, $\sigma$ and $\tau$ act only on 1-marked edges or cycles, respectively, of a graph $\Gamma$, hence are completely independent of its topology. On the other hand, $s$ and $t$ generate new 2-markings, so they depend on the incidence structure and the marking of $\Gamma$ in a non-trivial way. Nevertheless, it is possible to construct a universal model for all cases of (admissible) markings; in this section we show that every complex $\m P(\Gamma)$ can be modeled by marking the vertices of an associated graph $\Gamma'$. 

\begin{defn}\label{definition:Vcomplex}
Given a (not necessarily connected or 3-regular) graph $\Gamma$ (without external legs and self-loops) let $\m V(\Gamma)$ denote the free abelian group generated by all markings $(\Gamma,m)$ where $m:V \to \{0,1,2\}$ marks the (ordered set of) vertices of $\Gamma$ such that no two marked vertices are connected by an edge.

We write $V=V_0 \sqcup V_m$ with $V_m=V_1 \sqcup V_2$ for the partition of $V$ induced by a marking.
\end{defn}

Let $\m V(\Gamma)_i^j$ denote the subgroup generated by marked graphs with $i$ vertices of type 1 and $j$ vertices of type 2 and let
\begin{equation*}
 \m V(\Gamma)^j:= \bigoplus_{i \in \mb N} \m V(\Gamma)_i^j.
\end{equation*}

Mimicking the previous constructions (only the notion of admissible markings has changed) we obtain three differentials $\mu,u, U: \m V(\Gamma)^j \longrightarrow \m V(\Gamma)^{j+1}$, given by
 \begin{align*}
 \mu(\Gamma,m)&:= (-1)^{| V_m |}  \sum_{v \in V_1 } (-1)^{|\{ v' \in V_1  \mid v' > v\}|} \mu_v(\Gamma,m), \\
  u(\Gamma,m)&:= \sum_{v \in V_0 } (-1)^{|\{ v' \in V_m \mid v' < v\}|} u_v(\Gamma,m),\\ 
  U & :=u + \mu,
 \end{align*}
 where $\mu_v(\Gamma,m):= (\Gamma,m_{|v\mapsto 2})$ and 
 \begin{equation*}
  u_v(\Gamma,m):= \begin{cases}
           0 & \text{ if $v$ is adjacent to another marked vertex} \\
           (\Gamma,m_{|v\mapsto 2}) & \text{ else.}
          \end{cases}
 \end{equation*}
          
\begin{example}\label{example:vertexmarking}
Let $(\Gamma,m)=$
\raisebox{-.4cm}{\resizebox{1.4cm}{.9cm}{
\begin{tikzpicture}[scale=0.666]
  \coordinate  (v1) at (-1,1); 
   \coordinate  (v2) at (-1,-1);
   \coordinate (v3) at (0,0); 
   \coordinate  (v4) at (1,1); 
   \coordinate  (v5) at (1,-1);
  \fill[] (v1) circle (.066cm) node[xshift=-.3cm] {1};
   \fill[] (v2) circle (.066cm) node[xshift=-.3cm] {2};
\fill[] (v3) circle (.066cm) node[yshift=.3cm] {3};
\fill[] (v4) circle (.066cm) node[xshift=.3cm] {4};
\path node at (v4) [shape=circle,draw] {};
\fill[] (v5) circle (.066cm) node[xshift=.3cm] {5};
   \draw (v1) -- (v2) ;
   \draw (v1) -- (v3) ;
   \draw (v2) -- (v3) ;
   \draw (v3) -- (v4) ;
   \draw (v3) -- (v5) ;
   \end{tikzpicture}}}
with $V$ ordered as pictured, the 1- and 2-markings drawn as ``\mbox{\large$\circ$}'' and ``\mbox{\tiny $\square $}'', respectively. Then

\begin{align*}
 u(\Gamma,m) & = 
 \raisebox{-.33cm}{\resizebox{1.3cm}{.8cm}{
\begin{tikzpicture}[scale=0.666]
  \coordinate  (v1) at (-1,1); 
   \coordinate  (v2) at (-1,-1);
   \coordinate (v3) at (0,0); 
   \coordinate  (v4) at (1,1); 
   \coordinate  (v5) at (1,-1);   
  \fill[] (v1) circle (.066cm) node[xshift=-.3cm] {};
   \fill[] (v2) circle (.066cm) node[xshift=-.3cm] {};
\fill[] (v3) circle (.066cm) node[yshift=.3cm] {};
\fill[] (v4) circle (.066cm) node[xshift=.3cm] {};
\fill[] (v5) circle (.066cm) node[xshift=.3cm] {};
\path node at (v1) [shape=rectangle,draw] {};
\path node at (v4) [shape=circle,draw] {};
   \draw (v1) -- (v2) ;
   \draw (v1) -- (v3) ;
   \draw (v2) -- (v3) ;
   \draw (v3) -- (v4) ;
   \draw (v3) -- (v5) ;
   \end{tikzpicture}}}
   +
    \raisebox{-.33cm}{\resizebox{1.3cm}{.8cm}{
\begin{tikzpicture}[scale=0.666]
  \coordinate  (v1) at (-1,1); 
   \coordinate  (v2) at (-1,-1);
   \coordinate (v3) at (0,0); 
   \coordinate  (v4) at (1,1); 
   \coordinate  (v5) at (1,-1);   
  \fill[] (v1) circle (.066cm) node[xshift=-.3cm] {};
   \fill[] (v2) circle (.066cm) node[xshift=-.3cm] {};
\fill[] (v3) circle (.066cm) node[yshift=.3cm] {};
\fill[] (v4) circle (.066cm) node[xshift=.3cm] {};
\fill[] (v5) circle (.066cm) node[xshift=.3cm] {};
\path node at (v2) [shape=rectangle,draw] {};
\path node at (v4) [shape=circle,draw] {};
   \draw (v1) -- (v2) ;
   \draw (v1) -- (v3) ;
   \draw (v2) -- (v3) ;
   \draw (v3) -- (v4) ;
   \draw (v3) -- (v5) ;
   \end{tikzpicture}}}
   -
    \raisebox{-.33cm}{\resizebox{1.3cm}{.8cm}{
\begin{tikzpicture}[scale=0.666]
  \coordinate  (v1) at (-1,1); 
   \coordinate  (v2) at (-1,-1);
   \coordinate (v3) at (0,0); 
   \coordinate  (v4) at (1,1); 
   \coordinate  (v5) at (1,-1);   
  \fill[] (v1) circle (.066cm) node[xshift=-.3cm] {};
   \fill[] (v2) circle (.066cm) node[xshift=-.3cm] {};
\fill[] (v3) circle (.066cm) node[yshift=.3cm] {};
\fill[] (v4) circle (.066cm) node[xshift=.3cm] {};
\fill[] (v5) circle (.066cm) node[xshift=.3cm] {}; 
\path node at (v4) [shape=circle,draw] {};
\path node at (v5) [shape=rectangle,draw] {};
   \draw (v1) -- (v2) ;
   \draw (v1) -- (v3) ;
   \draw (v2) -- (v3) ;
   \draw (v3) -- (v4) ;
   \draw (v3) -- (v5) ;
   \end{tikzpicture}}}
   , \quad 
   \mu (\Gamma,m) = 
   -
    \raisebox{-.33cm}{\resizebox{1.3cm}{.8cm}{
\begin{tikzpicture}[scale=0.666]
  \coordinate  (v1) at (-1,1); 
   \coordinate  (v2) at (-1,-1);
   \coordinate (v3) at (0,0); 
   \coordinate  (v4) at (1,1); 
   \coordinate  (v5) at (1,-1);   
  \fill[] (v1) circle (.066cm) node[xshift=-.3cm] {};
   \fill[] (v2) circle (.066cm) node[xshift=-.3cm] {};
\fill[] (v3) circle (.066cm) node[yshift=.3cm] {};
\fill[] (v4) circle (.066cm) node[xshift=.3cm] {};
\fill[] (v5) circle (.066cm) node[xshift=.3cm] {}; 
\path node at (v4) [shape=rectangle,draw] {};
   \draw (v1) -- (v2) ;
   \draw (v1) -- (v3) ;
   \draw (v2) -- (v3) ;
   \draw (v3) -- (v4) ;
   \draw (v3) -- (v5) ;
   \end{tikzpicture}}}
   ,
   \\
   \mu u (\Gamma,m) & = 
    \raisebox{-.33cm}{\resizebox{1.3cm}{.8cm}{
\begin{tikzpicture}[scale=0.666]
  \coordinate  (v1) at (-1,1); 
   \coordinate  (v2) at (-1,-1);
   \coordinate (v3) at (0,0); 
   \coordinate  (v4) at (1,1); 
   \coordinate  (v5) at (1,-1);   
  \fill[] (v1) circle (.066cm) node[xshift=-.3cm] {};
   \fill[] (v2) circle (.066cm) node[xshift=-.3cm] {};
\fill[] (v3) circle (.066cm) node[yshift=.3cm] {};
\fill[] (v4) circle (.066cm) node[xshift=.3cm] {};
\fill[] (v5) circle (.066cm) node[xshift=.3cm] {};
\path node at (v1) [shape=rectangle,draw] {};
\path node at (v4) [shape=rectangle,draw] {};
   \draw (v1) -- (v2) ;
   \draw (v1) -- (v3) ;
   \draw (v2) -- (v3) ;
   \draw (v3) -- (v4) ;
   \draw (v3) -- (v5) ;
   \end{tikzpicture}}}
   +
    \raisebox{-.33cm}{\resizebox{1.3cm}{.8cm}{
\begin{tikzpicture}[scale=0.666]
  \coordinate  (v1) at (-1,1); 
   \coordinate  (v2) at (-1,-1);
   \coordinate (v3) at (0,0); 
   \coordinate  (v4) at (1,1); 
   \coordinate  (v5) at (1,-1);   
  \fill[] (v1) circle (.066cm) node[xshift=-.3cm] {};
   \fill[] (v2) circle (.066cm) node[xshift=-.3cm] {};
\fill[] (v3) circle (.066cm) node[yshift=.3cm] {};
\fill[] (v4) circle (.066cm) node[xshift=.3cm] {};
\fill[] (v5) circle (.066cm) node[xshift=.3cm] {};
\path node at (v2) [shape=rectangle,draw] {};
\path node at (v4) [shape=rectangle,draw] {};
   \draw (v1) -- (v2) ;
   \draw (v1) -- (v3) ;
   \draw (v2) -- (v3) ;
   \draw (v3) -- (v4) ;
   \draw (v3) -- (v5) ;
   \end{tikzpicture}}}
   -
    \raisebox{-.33cm}{\resizebox{1.3cm}{.8cm}{
\begin{tikzpicture}[scale=0.666]
  \coordinate  (v1) at (-1,1); 
   \coordinate  (v2) at (-1,-1);
   \coordinate (v3) at (0,0); 
   \coordinate  (v4) at (1,1); 
   \coordinate  (v5) at (1,-1);   
  \fill[] (v1) circle (.066cm) node[xshift=-.3cm] {};
   \fill[] (v2) circle (.066cm) node[xshift=-.3cm] {};
\fill[] (v3) circle (.066cm) node[yshift=.3cm] {};
\fill[] (v4) circle (.066cm) node[xshift=.3cm] {};
\fill[] (v5) circle (.066cm) node[xshift=.3cm] {}; 
\path node at (v4) [shape=rectangle,draw] {};
\path node at (v5) [shape=rectangle,draw] {};
   \draw (v1) -- (v2) ;
   \draw (v1) -- (v3) ;
   \draw (v2) -- (v3) ;
   \draw (v3) -- (v4) ;
   \draw (v3) -- (v5) ;
   \end{tikzpicture}}}
   = - u \mu (\Gamma,m).
\end{align*}
\end{example}

Adapting the proof of Proposition \ref{proposition:Ddifferential} we conclude that $(\m V(\Gamma),U)$ is a cochain complex. The universality of this complex is established by 
\begin{thm}\label{theorem:universal}
 Given a graph $\Gamma$ and $P\subset \mathrm{Sub}(\Gamma)$ let $(\m P(\Gamma),D)$ be the associated complex constructed in Section \ref{subsection:markedcomplex}. 
 Define a graph $\Gamma'=(V',E')$ by 
 \begin{equation*}
  V':=P, \quad E':= \{ (p,p') \mid \text{$p$ and $p'$ share a common vertex} \}\subset V'\times V'.
 \end{equation*}
 
Then $(\m P(\Gamma),D) \cong (\m V(\Gamma'),U)$ as cochain complexes.
\end{thm}

\begin{proof}
 Let the order on $V'$ be induced by the one on $P$. We define a linear map $\Psi:\m P(\Gamma)^j_i\to \m V(\Gamma')^j_i$ by $\Psi(\Gamma,m):=(\Gamma',m')$ with $m'(v):=m(p)$.  The definition of $E'$ implies that $m'$ is admissible if and only if $m$ is.  Then
 \begin{align*}
\Psi \delta (\Gamma,m)& =  (-1)^{| P_m |} \sum_{p \in P_1 }   (-1)^{|\{ p' \in P_1  \mid p' > p\}|}\Psi (\Gamma,m_{p\mapsto2})  \\
&= (-1)^{|P_m|} \sum_{v \in P_1 } (-1)^{|\{ p' \in P_1  \mid p' > p\}|} (\Gamma',m'_{v\mapsto 2}) \\
& =(-1)^{|V_m|}  \sum_{v \in V_1 } (-1)^{|\{ v' \in V_1  \mid v' > v\}|}(\Gamma',m'_{v\mapsto 2}) = \mu \Psi (\Gamma,m)
 \end{align*}
 and similar for $\Psi d=u \Psi$. Furthermore, $\Psi$ is invertible; its inverse is given by sending a marking $m'$ of $\Gamma'$ to the marking $m:P \to \{0,1,2\}, p \mapsto m'(v)$ where $v\in V'$ is the unique vertex representing $p$. 
\end{proof}

\begin{example}\label{example:lineandcyclegraphs}
 For the marked graphs in examples \ref{example:edgemarking} and \ref{example:cyclemarking} the associated graphs $(\Gamma',m')$, with the notation of Example \ref{example:vertexmarking}, are
\raisebox{-.2cm}{\resizebox{1.2cm}{.7cm}{
\begin{tikzpicture}[scale=0.666]
  \coordinate  (v1) at (-1,1); 
   \coordinate  (v2) at (-1,-1);
   \coordinate (v3) at (0,0); 
   \coordinate  (v4) at (1,1); 
   \coordinate  (v5) at (1,-1);
  \fill[] (v1) circle (.066cm) node[xshift=-.3cm] {1};
   \fill[] (v2) circle (.066cm) node[xshift=-.3cm] {2};
\fill[] (v3) circle (.066cm) node[yshift=.3cm] {3};
\fill[] (v4) circle (.066cm) node[xshift=.3cm] {4};
\fill[] (v5) circle (.066cm) node[xshift=.3cm] {5};
\path node at (v1) [shape=circle,draw] {};
   \draw (v1) -- (v2) ;
   \draw (v1) -- (v3) ;
   \draw (v2) -- (v3) ;
   \draw (v3) -- (v4) ;
   \draw (v3) -- (v5) ;
   \draw (v4) -- (v5) ;
   \end{tikzpicture}}}
and \raisebox{-.05cm}{\resizebox{1.2cm}{.3cm}{
\begin{tikzpicture}[scale=0.666]
  \coordinate  (v1) at (-1,0); 
   \coordinate  (v2) at (0,0);
  \fill[] (v1) circle (.066cm) node[xshift=-.3cm] {1};
   \fill[] (v2) circle (.066cm) node[xshift=.3cm] {2};
\path node at (v2) [shape=circle,draw] {};
   \end{tikzpicture}}}, respectively. 
   The complex $(\m V(\Gamma'),U)$ contains thus all information of $(\m P(\Gamma),D)$ while being as simple as possible.
\end{example}

\begin{remark}
We rephrase the content of this section in abstract terms.  Consider a category of complexes of marked graphs, formed by complexes of the form $(\m P(\Gamma),D)$ for $\Gamma$ a finite graph and cochain maps between them.
Then we have shown that vertex-markings form a subcategory which is equivalent to the larger category; the inclusion of this subcategory defines a functor which is fully-faithful by construction and the previous theorem shows that it is also essentially surjective, hence part of an equivalence.
\end{remark}

\section{Cohomology of complexes of marked graphs}\label{section:ggcohom}

In this section we show that both the edge- and cycle-marking complexes are acyclic. By Theorem \ref{theorem:universal} it suffices to consider the case of vertex-markings, i.e.\ the complex $(\m V,U)$ where $\m V:=\bigoplus_{\Gamma \in \mathrm{Gra}_{r,l}}\m V(\Gamma')$. Moreover, since taking homology commutes with taking direct sums, we focus on the complexes $(\m V(\Gamma), U)$ for an arbitrary graph $\Gamma$.

The proof is based on two steps. First we calculate the cohomology of the complex $(\m V(\Gamma),\mu)$, then we use the result in a spectral sequence to find the cohomology of the double complex $(\m V(\Gamma),u+\mu)$.

\subsection{The cohomology of the complex $(\m V,\mu)$}\label{subsection:hommu}
Fix a graph $\Gamma$ and let $(\Gamma,m) \in \m V(\Gamma)_i^j$, so $m$ specifies $i$ 1- and $j$ 2-marked vertices. Since the map $\mu$ changes 1-markings into 2-markings, the image $\mu(\Gamma,m)$ is an element of $\m V(\Gamma)_{i-1}^{j+1}$, given by the sum over $i$ copies of $\Gamma$ where a single 1-marked vertex has been replaced with a 2-marked vertex. 

We can model this situation in a rather simple way. Let $\Lambda_n$ denote the graph on $n$ disconnected vertices, ordered by $v_1< \ldots <v_n$, and define a chain complex\footnote{In this section we use homological conventions, i.e.\ we grade by the number of 1-marked elements. The connection to the cohomology of $(\m V(\Gamma),\mu)$ is established at the end.} $(\m L(\Lambda_n),\mu)$ by 
\begin{equation*}
\m L(\Lambda_n)_i:= \langle (\Lambda_n,m)\mid m:V(\Lambda_n) \to \{1,2\}, |m^{-1}(1)|=i \rangle. 
\end{equation*}
Note that here the markings are required to mark every vertex of $\Lambda_n$. This complex captures the action of $\mu$ on a single configuration of marked vertices, i.e.\ on markings $m, m'$ with $V_m=V_{m'}$ as subsets of $V$. There are however also other configurations of marked vertices which have to be taken into account. 

\begin{defn}
An \textit{independent set} of size $n$ in a graph $\Gamma$ is a subset of $V$ of size $n$ such that no two of its elements are adjacent. We write $I_n=I_n(\Gamma)\subset 2^{2^V}$ for the set of all independent sets with size $n$ in $\Gamma$.
\end{defn}

Rephrasing Definition \ref{definition:Vcomplex} we see that every marking $m$ on $\Gamma$ corresponds bijectively to an independent set $V_m\subset V$ with a choice of labeling or 2-partition $V_m=V_1 \sqcup V_2$ of its elements. Taking as many copies of $\m L(\Lambda_n)$ as there are independent sets of size $n$ in $\Gamma$ allows to model the action of $\mu$ on $\m V(\Gamma)$.

\begin{lem}
Given $\Gamma \in \mathrm{Gra}_{r,l}$ define a chain complex $(\m L,\mu)$ by 
\begin{equation}\label{eq:cochainmap}
 \m L_i:=\bigoplus_{j \in \mb N} \bigoplus_{m \in I_{i+j}(\Gamma) }\m L(\Lambda_{i+j})_i. 
\end{equation}
Then there is an isomorphism of chain complexes $\left( \m V(\Gamma),\mu \right) \cong \left( \m L,\mu \right)$.
\end{lem}

\begin{proof}
Recall that the vertices of $\Gamma$ are ordered. Thus, for each marking $m$ on $\Gamma$ with $|V_m|=i+j$ there is an induced order on $V_m(\Gamma)$ and a unique order preserving bijection $\varphi_m$ between $V_m(\Gamma)$ and the vertices of $\Lambda_{i+j}$ (which are all marked). For any $m$ in $I_{i+j}(\Gamma)$ we send $(\Gamma,m) \in \m V(\Gamma)_i^j$ to (a copy of) $(\Lambda_{i+j},m')$ where $m'=m \circ \varphi_m^{-1}$. This defines a chain map from $\m V(\Gamma)_i$ to $\m L_i$ which is bijective by construction.  
\end{proof}

\begin{lem}
 $H_\bullet(\m L(\Lambda_n),\mu)=0$ for all $n>0$.
\end{lem}

\begin{proof}
Fix $n>0$. To prove the lemma, we define a chain isomorphism $\Phi$ between $(\m L(\Lambda_n),\mu)$ and the augmented (and degree shifted) simplicial chain complex $\ti C(\Delta)[+1]$ of the standard $(n-1)$-simplex $\Delta=[v_1,\ldots,v_n]$ which is contractible, hence has vanishing reduced homology.
 
For $\ti C_i:=C_{i-1}(\Delta), \ti C_0=\mb Z$ and $\partial_0=\varepsilon: \sum \lambda_iv_i \mapsto \sum \lambda_i$ define 
 \begin{equation*}
\Phi_i: \m L(\Lambda_n)_i \to \ti C_i, \quad  \Phi_i(\Lambda_n,m):= \begin{cases}
                         \left[\{ v_k \mid v_k \text{ is 1-marked} \}\right] & i>0 \\
                         1 & i=0.
                        \end{cases}
 \end{equation*}
 
 Let $m_v$ mark a single vertex by 1, then we compute 
 \begin{equation*}
  \varepsilon \Phi_1(\Lambda_n,m_v)= \varepsilon[v]=1= \Phi_0 \mu (\Lambda_n,m\equiv 2)
 \end{equation*}
 where $m\equiv 2$ denotes the constant map, marking every vertex by 2.
 
For $l>0$ and $v_{k_1}< \ldots < v_{k_l}$ the 1-marked vertices of $\Lambda_n$
\begin{equation*}
 \partial \Phi_l(\Lambda_n,m)= \sum_{i=1}^{l} (-1)^{i+1}[v_{k_1}, \ldots,\overset{\wedge}v_{k_i}, \ldots, v_{k_l}]
\end{equation*}
and
 \begin{align*}
 \Phi_{l-1}\mu(\Lambda_n,m) &=   (-1)^n  \sum_{i=1}^l (-1)^{l-i} \Phi (\Lambda_n,m_{|e_i\mapsto 2}) \\
 & = (-1)^n  \sum_{i=1}^l (-1)^{l-i} [v_{k_1}, \ldots,\overset{\wedge}v_{k_i}, \ldots, v_{k_l}].
\end{align*}
The expressions match up to a sign $(-1)^{n-l+1}$ which may be absorbed by $\partial\mapsto (-1)^n \partial$ without changing the homology of this complex. The map $\Phi$ is clearly bijective, so it induces an isomorphism on homology.
\end{proof}

There is one summand on the right side of \eqref{eq:cochainmap} which has non-trivial homology, the piece corresponding to the empty graph $\Lambda_0$ representing the case where $\Gamma$ has no marked vertices at all. This element is $\mu$-closed, but not exact. 
\begin{lem}
 $H_k(\m L(\Lambda_0),\mu)=0$ for all $k>0$ and isomorphic to $\mb Z$ in degree 0.
\end{lem}
\begin{proof}
 $\m L(\Lambda_0)$ consists of a single element, the empty graph, concentrated in degree 0, and $\mu$ maps it to 0.
\end{proof}

Putting everything together we arrive at
\begin{prop}\label{proposition:homsigma}
 The homology of the complex $(\m V(\Gamma),\mu)$ is given by
\begin{equation*}
 H_k(\m V(\Gamma),\mu)  \cong \begin{cases}
                     \mb Z & k=0\\
                      0 & \text{else.}
                     \end{cases}
\end{equation*}
\end{prop}

Finally, since $\mu$ is of bidgree $(-1,+1)$ we have
\begin{equation*}
 H_k\Big(\bigoplus_{j \in \mb N} \m V(\Gamma)_\bullet^{j-\bullet},\mu \Big) =  \bigoplus_{j \in \mb N} H_k \big(\m V(\Gamma)_\bullet^{j-\bullet},\mu \big) = \bigoplus_{j \in \mb N}    H^{j-k}\big( \m V(\Gamma)^\bullet_k,\mu \big),
\end{equation*}
so that 
\begin{equation*}
 H^n(\m V(\Gamma),\mu)  \cong \begin{cases}
                     \mb Z & n=0,\\
                      0 & \text{else,}
                     \end{cases}
\end{equation*}
and the sole class in $H^0(\m V(\Gamma),\mu)$ is represented by $(\Gamma,m_0)\in \m V(\Gamma)_0^0$, the graph $\Gamma$ with no marked vertices.

\subsection{The cohomology of the complex $(\m V,u+\mu)$}\label{subsection:homuplusmu}

Having understood the cohomology of $\m V$ with respect to $\mu$, we now consider the full differential $U=u+\mu$. Again, it suffices to study each summand $\m V(\Gamma)$ individually. The bigrading is given by the number of 1- and 2-marked vertices so that $u$ and $\mu$ have bidegrees $(0,1)$ and $(-1,1)$. 
In the following it will be convenient to change this bigrading into the total number of marked vertices and those of type 1. From now on we work with cohomological grading, i.e.\ we take the second part of the bigrading to be negative.

Hence, given a graph $\Gamma$ we define
\begin{equation*}
\m T \!:= \! \m T(\Gamma), 
\m T^{i,j} \!  := \big\langle  (\Gamma,m) \mid m: V \to \{0,1,2\}, |V_m|=i,|V_1|=-j \big\rangle = \m V(\Gamma)^{i+j}_{-j}.
\end{equation*}
The differentials $u$ and $\mu$ are then of bidegree $(1,0)$ and $(0,1)$, respectively. The associated total complex is $(\m T,u+\mu)$ where $\m T^n= \bigoplus_{i+j=n} \m T^{i,j}$, so $n$ is the number of 2-marked vertices.

\begin{thm}\label{theorem:HomDC}
For any graph $\Gamma$ the double complex $(\m T,U)$ is acyclic, i.e.\ the cohomology $H^n(\m T,U)$ vanishes for all $n>0$ and is isomorphic to $\mb Z$ in degree 0.
\end{thm}

\begin{proof}
 It is possible to give a constructive proof making the isomorphism explicit, but we opt for a spectral sequence argument. Our notation follows the conventions in \cite{gm-homalg}.
 
 Let $\m T$ be filtered by $\m T=F^0 \m T\supseteq \ldots \supseteq F^p \m T\supseteq \ldots \supseteq F^m\m T= 0$ with
 \begin{equation*}
  F^p\m T^n:=\bigoplus_{i+j=n, i\geq p } \m T^{i,j}.
 \end{equation*}
The associated spectral sequence starts with
\begin{align*}
 &E_0^{p,q}= F^{p} \m T^{p+q} / F^{p+1} \m T^{p+q} = \m T^{p,q}, \\
 &d_0^{p,q}:E_0^{p,q}\longrightarrow E_0^{p,q+1}= \mu:\m T^{p,q}\longrightarrow \m T^{p,q+1}.
\end{align*}
On its first page we have $E_1^{p,q}=H^q(\m T^{p,\bullet},\mu)$ and $d_1^{p,q}$ induced by $u$. But according to Proposition \ref{proposition:homsigma} the only non-zero term is $E_1^{0,0}=H^0(\m T^{0,\bullet},\mu)$.
All the maps $d_1^{p,q}$ are thus zero, so that the sequence collapses at its first page and we have
\begin{equation*}
 E_\infty^{p,q}= G^pH^{p+q}(\m T,u+\mu)=\begin{cases}
                                          H^0(\m T^{0,\bullet},\mu) & p,q=0 \\
                                          0 & \text{else.}
                                         \end{cases}
\end{equation*}

Thus, $H^0(\m T,U)\cong \mb Z$ and all other cohomology groups of $(\m T, U)$ are trivial.
\end{proof}

\begin{remark}
As mentioned in the introduction, the cohomology with respect to the differential $u$ is highly non-trivial (cf.\ \cite{knispel}). In principle, it could be studied by the methods of \cite{kwc-dogc}: Set up a spectral sequence for $(\m T,u+\mu)$ as above, but filtered ``in the other direction", i.e.\ with $E_1^{p,q}=H^p(\m T^{\bullet,q},u)$. Since the total complex $H^\bullet(\m T,u+\mu)$ is acyclic, classes in $H^\bullet(\m T,u)$ must have ``partners" which they kill or get killed by on some later page of the spectral sequence.
\end{remark}

\subsection{The cohomology of the edge- and cycle-marking complexes}\label{subsection:cohomgagcc}
In terms of the edge- and cycle-marking complexes Theorem \ref{theorem:HomDC} translates into
\begin{equation*}
 H^k(\m E, S)= \bigoplus_{\Gamma \in \mathrm{Gra}_{r,l}}  H^k(\m E(\Gamma),S)\cong \bigoplus_{\Gamma \in \mathrm{Gra}_{r,l}} H^k(\m V(\Gamma'),U) \cong\begin{cases}
                                          \bigoplus_{\Gamma}\mb Z & k=0 \\
                                          0 & \text{else}                                                                                                                                                                                                                  \end{cases}
\end{equation*}
and similar for $(\m C,T)$.

It is possible to write down explicit maximal generators for these cohomology groups.
\begin{prop}\label{prop:cocycles}
Define two linear maps $\chi_+:\m E \to \m E$ and $\delta_+:\m C \to \m C$ by 
\begin{align*} 
\chi_+ (\Gamma,m) &:= \sum_{e \in E} \chi_+^e (\Gamma,m), \\
 \chi_+^e(\Gamma,m) &:= \begin{cases} 
0 & \text{if $e$ is adjacent to another marked edge} \\
(\Gamma,m_{e \mapsto 1}) & \text{else,} 
                        \end{cases}
\end{align*}
and
\begin{align*}
\delta_+ (\Gamma,m) & := \sum_{c \in C} \delta_+^c (\Gamma,m), \\
\delta_+^c(\Gamma,m) & := \begin{cases} 
0 & \text{if $c$ is adjacent to a marked edge} \\
(\Gamma,m_{c \mapsto 1}) & \text{else.} 
                        \end{cases}
\end{align*}
Denote by $m_0: E\to \{0\}$ the trivial marking. Then $Se^{\chi_+} (\Gamma,m_0)=0$ and $Te^{\delta_+} (\Gamma,m_0)=0$.
\end{prop}

\begin{proof}
 A straightforward calculation, see propositions 4.29 and 4.35 in \cite{ksvs}.
\end{proof}

From a physics point of view, as advocated in \cite{ksvs}, this establishes the expected result; the cohomology of $(\m E,S)$ or $(\m C,T)$ is generated by the sum over all graphs marked accordingly which resembles the pure gluon or gluon/ghost amplitudes\footnote{Note that ``real'' ghost cycles are directed. This can be included without changing any of the results by declaring a marked cycle to represent the sum of two ghost cycles with opposite orientation.}.

\section{The total complex}\label{section:gaugegraphcomplex}
We can combine the edge- and cycle-marking graph complexes into a single complex by considering admissible markings $m:E\sqcup C \to \{0,1,2\}$. 

Define the free abelian groups
 \begin{equation*}
  \m G^{i,j}:= \left\langle (\Gamma,m) \mid \Gamma \in \mathrm{Gra}_{r,l}, m:E\sqcup C \to \{0,1,2\}:|E_2|=i, |C_2|=j\right\rangle.
 \end{equation*}
These groups populate a double complex with horizontal differential $S$ and vertical differential $T$. The next lemma allows us to form the associated total complex,
\begin{equation*}
 \m G^n:= \bigoplus_{i+j=n}\m G^{i,j}, \quad S+(-1)^nT:\m G^n \longrightarrow \m G^{n+1}.
\end{equation*}

 \begin{lem}
  The maps $S$ and $T$ commute.
 \end{lem}

 \begin{proof}
  Since $\tau$ and $\sigma$ change the respective 1-markings only, we have $\tau \sigma =\sigma \tau, \tau s=s\tau$ and $t \sigma =\sigma t$ by definition. 
  
  It remains to show that $t$ and $s$ commute. Given a marked graph $(\Gamma,m)$ write $E'\subset E_0$ for the edges of $\Gamma$ that can be 2-marked and for each $e \in E'$ let $C'_e \subset C_0$ be the set of cycles disjoint from $e$ and any other marked edge or cycle. Then
   \begin{align*}
  ts (\Gamma,m) = & \sum_{c \in C_0 } \sum_{e \in E_0 }(-1)^{|\{ c' \in C_m \mid c' < c\}|+|\{ e' \in E_m \mid e' < e\}|} t_c s_e(\Gamma,m) \\
  = &\sum_{e \in E' } \sum_{c \in C'_e }(-1)^{|\{ c' \in C_m \mid c' < c\}|+|\{ e' \in E_m \mid e' < e\}}(\Gamma,m_{e\mapsto 2,c\mapsto 2}),
  \end{align*}
  all other summands vanish.
  The expression for $st (\Gamma,m)$ is similar,
   \begin{equation*}
   s t (\Gamma,m) = \sum_{c \in C'' } \sum_{e \in E_c'' }(-1)^{|\{ c' \in C_m \mid c' < c\}|+|\{ e' \in E_m \mid e' < e\}|} (\Gamma,m_{e\mapsto 2,c\mapsto 2})
  \end{equation*}
  with $C''\subset C_0$ and $E_c'' \subset E_0$ defined as above. Thus, $ts(\Gamma,m)$ and $st(\Gamma,m)$ are both given by sums over subsets of $E_0 \times C_0$,
  \begin{equation*}
   \{(e,c) \mid e \in E' \wedge c\in C_e'\} \text{ and } \{(e,c) \mid c \in C'' \wedge e\in E_c''\}.
  \end{equation*}
If a pair $(e_*,c_*)$ is in the left hand set, then $e_*$ and $c_*$ can both be 2-marked in $\Gamma$ at the same time, i.e.\ the order is unimportant. Therefore, the pair is also in the right hand set. Hence both sums are identical, so $[s,t]=0$.
 \end{proof}

\begin{thm}
 The cohomology of the total complex is given by
 \begin{equation*}
  H^n\big(\m G^\bullet, S+(-1)^\bullet T \big)\cong \begin{cases}
                    \bigoplus_{\Gamma \in \mathrm{Gra}_{r,l}} \mb Z & \text{if }n=0 \\
                    0 & \text{else.}
                   \end{cases}
 \end{equation*}
\end{thm}

\begin{proof}
  Same as the proof of Theorem \ref{theorem:HomDC}.
\end{proof}

 Moreover, an explicit maximal generator is given by the following construction. Let $m_0:E\sqcup C \to \{0\}$ denote the trivial marking. If
 \begin{equation*}
 X_{r,l}:=\sum_{\Gamma \in \mathrm{Gra}_{r,l}} (\Gamma,m_0)
 \end{equation*}
 encodes the $l$-th order $r$-point amplitude for a scalar cubical field, then its general gauge theoretic counterpart\footnote{Including fermions into the picture is a mere technicality. The compatibility of symmetry factors however is not so simple. See \cite{ksvs}, Lemma 4.10 and 4.24, as well as \cite{kiss-diana} and \cite{sars}.} is
  \begin{equation*}
\ti X_{r,l}:=\sum_{\Gamma \in \mathrm{Gra}_{r,l}} e^{\delta_+}e^{\chi_+} (\Gamma,m_0),
 \end{equation*}
 the sum over all 1-marked graphs. Proposition \ref{prop:cocycles} shows that $S\ti X_{r,l}=0$ and $T\ti X_{r,l}=0$, so $\ti X_{r,l}$ is the sole maximal generator of $H^0(\m G^\bullet, S+(-1)^\bullet T )$. Putting everything together we have proven Theorem \ref{thm:mainthm}.

\bibliography{ref.bib}   
\bibliographystyle{alpha}  

\end{document}